\documentclass[sigconf, nonacm]{acmart}

\setcopyright{none}
\copyrightyear{2021}
\acmYear{2021}
\acmDOI{}
\acmConference[]{}{}{}
\acmBooktitle{}
\acmPrice{}
\acmISBN{}

\usepackage[utf8]{inputenc}
\usepackage{color}
\usepackage{amsmath,amsthm}
\usepackage{geometry}[margin=1in]
\usepackage{makecell}
\usepackage{hyperref}

\usepackage[ruled,vlined]{algorithm2e}

\newtheorem{theorem}{Theorem}[section]
\newtheorem{lemma}[theorem]{Lemma}

\newtheorem{corollary}[theorem]{Corollary}
\newtheorem{definition}[theorem]{Definition}
\newenvironment{proofsketch}{\begin{proof}[Proof Sketch]}{\end{proof}}

\begin{document}

\title{Contention Resolution with Predictions}

\author{Seth Gilbert}
\email{seth.gilbert@comp.nus.edu.sg}
\affiliation{
    \institution{National University of Singapore}
    \country{Singapore}
    }

\author{Calvin Newport}
\email{cnewport@cs.georgetown.edu}
\affiliation{
    \institution{Georgetown University}
    \country{United States}
    }

\author{Nitin Vaidya}
\email{nitin.vaidya@georgetown.edu}
\affiliation{
    \institution{Georgetown University}
    \country{United States}
    }
    
\author{Alex Weaver}
\email{aweaver@cs.georgetown.edu}
\affiliation{
    \institution{Georgetown University}
    \country{United States}
    }





\begin{abstract}
In this paper, we consider contention resolution algorithms that are augmented with predictions
about the network. We begin by studying the natural setup in which the algorithm is provided a distribution defined over the possible network sizes that predicts the likelihood of each size occurring. The goal is to leverage the predictive power of this distribution to improve on worst-case time complexity bounds. Using a novel connection between contention resolution and information theory,
we prove lower bounds on the expected time complexity with respect to the Shannon entropy
of the corresponding network size random variable, for both the collision detection and no collision detection
assumptions.
We then analyze upper bounds for these settings, 
assuming now that the distribution provided as input might differ from the actual distribution generating network sizes.
We express their performance with respect to both entropy
and the statistical divergence between the two distributions---allowing us to quantify the cost of poor predictions.
Finally, we turn our attention to the related perfect advice setting,
parameterized with a length $b\geq 0$, in which all active processes in a given execution
are provided the best possible $b$ bits of information about their network.
We provide tight bounds on the speed-up possible with respect to $b$ for deterministic
and randomized algorithms, with and without collision detection.
These bounds provide a fundamental limit on the maximum power that can be provided by
any predictive model with a bounded  output size.
\end{abstract}

\maketitle

\begin{acks}
This work was funded by Singapore Ministry of Education Grant MOE2018-T2-1-160 (Beyond Worst-Case Analysis: A Tale of Distributed Algorithms) and NSF CCF Awards 1733842 and 1733872 (AiTF: Collaborative Research: Algorithms for Smartphone Peer-to-Peer Networks).
\end{acks}

\section{Introduction}
\label{sec:introduction}

In this paper, we study distributed algorithms that leverage predictions about their environment to improve worst-case performance bounds.
Motivated by recent investigations in the context of sequential algorithms (e.g.,~\cite{mitzenmacher2020algorithms}),
we imagine these predictions might be generated in practice by machine learning models able to observe the behavior of a given
environment over time.
As Mitzenmacher and Vassilvitskii note in their recent study of the ski rental problem with predictions~\cite{mitzenmacher2020algorithms},
one appealing goal for this general approach is to produce algorithms that perform no worse than our current optimal
solutions, but that will subsequently improve ``for free'' as the machine learning models generating the predictions they
leverage improve in the future.

To help establish foundations for studying distributed algorithms with predictions,
we focus on the classical contention resolution problem, as it is both simple and well-studied.
We begin by producing new lower and upper bounds on the speed-up achievable when the algorithm is provided a predicted distribution over network sizes. 
To achieve the strongest results, our lower bounds assume the distribution accurately describes the likelihood of each network size occurring, while our upper bounds explicitly capture how this performance degrades with respect to the Kullback–Leibler (KL) divergence between the predicted distribution and the actual distribution.
We also study lower and upper bounds for a {\em perfect advice} setting in which all active processes are provided the same $b$ bits of 
advice generated by a function with perfect knowledge of the network. The goal here is to understand the maximum possible improvement for a given amount of information.

As we elaborate below, our network size prediction results are built on a novel connection between contention resolution and information theory, leading to speed-up results expressed with respect to the entropy of the size distribution, lower bounds that leverage
the classical Source Code Theorem for optimal coding on noiseless channels,
and upper bounds that use optimal codes in their design. A nice property of our information theory-based
lower bounds is that  when they are applied to distributions with high entropy, they match (or almost match) 
longstanding existing lower
bounds for contention resolution with uniform algorithms, highlighting yet another mathematical framework through which to understand
this fundamental problem.

\subsection{Contention Resolution with Network Size Predictions}
The contention resolution problem assumes an unknown and non-empty subset $P\subseteq V$ of players from $|V| = n$ possible
participants are activated and connected to a shared channel.
Time proceeds in synchronous rounds.
In each round, each player can choose to transmit or receive. If two or more players transmit, all messages in that round
are lost due to collision.
If the channel supports collision detection, all players (including the transmitters) detect a collision. If the channel does not support collision detection, then players detect silence.
The problem is solved in the first round in which exactly one player transmits.

In resolving contention on a shared channel, knowing the size $k$ of the participant set $P$ is useful.
Indeed, many of the standard optimal worst-case algorithms operate by efficiently trying to find a good
estimate of this size. 
The {\em decay} strategy for channels without collision detection, for example,
introduced by Bar Yehuda et al.~\cite{bgi}, 
can be understood as cycling through $\log{n}$ geometrically-distributed 
guesses of the network size.
Given collision detection,
Willard's algorithm~\cite{willard1986log} achieves optimal time complexity by conducting
a binary search over these same guesses.

In a worst-case setting in which no information is provided about the network size,
Jurdzi{\'n}ski and Stachowiak proved a lower bound 
of $\Omega(\log{n}/\log\log{n})$ on the expected rounds
to solve contention resolution with no collision detection and uniform algorithms (in which probabilities are fixed in advance)~\cite{jurdzinski2002probabilistic}.
Colton et al.~\cite{farach2006lower} later eliminated the $\log\log{n}$ factor,
and Newport~\cite{newport2014radio} subsequently generalized the bound to non-uniform algorithms.
These results  match the $O(\log{n})$
expected round complexity of the {\em decay} strategy~\cite{bgi}.
When collision detection is available,
Willard establishes
 $\Theta(\log\log{n})$ to be a tight bound~\cite{willard1986log}.

In both the collision detection and no collision detection settings,
if the algorithm is given an accurate estimate $\hat k = \Theta(k)$ of
the actual network size $k$,
the problem can be solved in $O(1)$ rounds in expectation by simply transmitting with probability $1/\hat k$ in each round.
In real networks, however, it may be difficult to predict with confidence the exact network size to expect next.
It is arguably more likely that learning models will generate {\em distributions} over the likelihood of various sizes.
The relevant question we explore in Section~\ref{sec:networksize}, therefore,
is how much prediction distributions of varying usefulness and quality allow us to speed-up from
the worst-case toward the best-case bounds.

In Section~\ref{sec:networksize}, we augment the standard contention resolution problem
to provide each algorithm as input the definition
of a random variable $X$ defined over the possible participant set sizes from $1$ to $n$.
The goal is to use the distribution defining $X$ to speed up contention resolution if possible.
(For concision, in the following we sometimes say the algorithm is provided 
a random variable over network sizes as input. In these instances, we mean that the algorithm is provided as input the full distribution defining the variable.)

As we argued earlier, however, we do not need the exact network size to solve contention
resolution fast. An estimate within a constant factor of the real size is sufficient.
Given a network size random variable $X$, let $c(X)$ be the {\em condensed} version of $X$ that aggregates the probability mass over
$\lceil \log{n} \rceil$ values ranges of geometrically increasing size.
Formally,
for each $i\in [\lceil \log{n} \rceil]$, define:

\[ \Pr\left(c(X) = i\right) = \sum_{2^{i-1} < j \leq 2^i} p_j. \]

Intuitively, knowing a value $i\in [\lceil \log{n} \rceil]$, 
such that $k = \Theta(2^i)$, is sufficient. 
We therefore express the bounds that follow for a given $X$ with respect to $c(X)$.

\begin{table}
\centering
\renewcommand{\arraystretch}{1.5}
\begin{tabular}{ | c | c | c |}
    \hline
  & \makecell{Lower \\ Bounds (*)} & \makecell{Upper \\ Bounds (*)} \\ 
 \hline
 \makecell{No Collision \\ Detection} & $\Omega\bigg(\frac{2^{H(c(X))}}{\log{\log{n}}}\bigg)$ & $O\Big(2^{2H(c(X))}\Big)$ \\  
 \hline
 \makecell{Collision \\ Detection} & $\frac{H(c(X))}{2}-O(\log{\log{\log{\log{n}}}})$ & $O(H^2(c(X)))$ \\
 \hline
\end{tabular} 
\renewcommand{\arraystretch}{1}

\caption{\label{tab:predictionresults} Here we summarize our results for contention resolution with network size predictions. $H(c(X))$ represents the entropy of the condensed probability distribution over the possible network sizes.\\
(*) Lower bounds are expressed as the expected number of rounds while upper bounds are expressed with respect the number of rounds with at least constant probability. }
\end{table}

\subsubsection*{Lower Bounds}

We begin with lower bounds. 
To achieve the strongest possible result, we assume that the variable $X$ is accurate,
in the sense that the network will actually determine the number of participants according to $X$.
We also assume {\em uniform} algorithms in which all players use the same transmission probability
in each round from a fixed schedule.\footnote{It is important to point out that until recent work~\cite{newport2014radio},
most existing lower bounds for contention resolution assumed uniform algorithms.}

Intuitively, if $c(X)$ places all of its probability mass on a single network size range, then we are in the perfect
prediction setting and can solve the problem in $O(1)$ rounds.
On the other extreme, if $c(X)$ describes a uniform distribution, we are likely unable to do much better
than the worst-case bounds.
To describe these cases we need a property of these distributions that succinctly captures their predictive power.
The natural candidate here is {\em entropy}---and this is indeed what we we end up deploying in our results.

We begin in Section~\ref{sec:predictions:lower:nocd} by consider contention resolution with no collision detection.
We prove that for a given network size random variable $X$,
contention resolution requires $\Omega\left( \frac{2^{H(c(X))}}{\log\log{n}} \right)$ rounds in expectation,
where $H$ denotes Shannon entropy.
For maximum values of $H(c(X))$, 
this reduces to $\Omega(\log{n}/\log\log{n})$ rounds,
exactly matching the original uniform algorithm lower bound from~\cite{jurdzinski2002probabilistic}---a result that is
interesting in its own right, as it shows yet another approach to proving limits to this fundamental problem.

The core idea driving this bound is a connection between contention resolution and coding.
Using a reduction argument involving an intermediate game we call range finding,
we formalize the intuition that solving contention resolution requires an algorithm to try a transmission probability
relatively close to the optimal probability.
We can therefore use a contention resolution algorithm designed for $X$ to help construct a code for a given symbol source $c(X)$.
Let $s\in [\lceil \log{n} \rceil]$ be the symbol to be transmitted.
We can then consider the execution of our algorithm in a network of size $2^s$,
and identify the first round $r$ such that the algorithm attempts a transmission probability near $2^{-s}$,
and send $r$ to the receiver.
Because the receiver knows the same algorithm, it can use the probability scheduled for round $r$ to determine
$s$. If the algorithm terminates fast in expectation with respect to $X$,
then the expected code length of our code for $c(X)$ is small.
The Source Code Theorem (see Section~\ref{sec:networksize:prelim}), however,
tells us that the average code length for $c(X)$ is lower bounded by $H(c(X))$.
Because, roughly speaking, the code length is computed as the logarithm of the  round complexity,
we get a bound of roughly $2^{c(H(X))}$ rounds.

In Section~\ref{sec:predictions:lower:cd},
we leverage a similar connection between contention resolution and information theory
to prove that given collision detection,
uniform algorithms that solve contention resolution with respect to a network size variable $X$ require
$\Omega(H(c(X)))$ in expectation. For the maximum possible entropy of $c(X)$,
this matches the original $\Omega(\log\log{n})$ lower bound proven
by Willard, using non-information theoretic techniques~\cite{willard1986log}.

A uniform algorithm in the collision detection setting can be understand as a function that maps
a binary collision history\footnote{In the collision detection model, in every round,
either every player detects a collision or no player detects a collision.
We can therefore encode this history for the $r$ rounds of an execution
as a binary string $b_1b_2\ldots b_r$, where $b_i = 0$ if there was no collision
in round $i$, and $b_i=1$ if there was a collision in $i$.} to the uniform probability that all players
try in the current round given the current collision history.
We show how to construct a code for $c(X)$ for such an algorithm by directly sending
the shortest collision history that corresponds to a probability well-matched to the current
symbol. (As in the no collision detection case, there is actually an intermediate step involved
here where our contention resolution algorithm reduces to a simpler game called range finding,
and it is in fact this range finding solution from which we generate our code.)
In this case, they round complexity is the same as the code length. It follows
that $H(c(X))$ must bound the expected value of the former.

\subsubsection*{Upper Bounds}

We next turn our attention to produce contention resolution upper bounds that
take a network size random variable as input and attempt to come as close as possible
to matching the relevant lower bounds.
We start in Section~\ref{sec:predictions:upper:nocd} with an algorithm for the no collision detection setting.
Let $X$ be the actual distribution from which the network size will be drawn.
Let $Y$ be the distribution the algorithm is given as input (e.g., the distribution learned).
We analyze a natural strategy: sort the values in $[\lceil \log{n} \rceil]$ in decreasing order of likelihood
given $c(Y)$; visit these values in turn, for each such $i$ transmitting withing probability $2^{-i}$.

We prove that with constant probability, this strategy succeeds in $O(2^T)$ rounds,
where:

\[T=2H(c(X)) + 2D_{KL}(c(X)\Vert c(Y)),\] 

\noindent and $D_{KL}$ denotes the Kullback–Leibler divergence between the two distributions
(see Section~\ref{sec:networksize:prelim}).
When $Y=X$, this simplifies to $O(2^{2H(c(X))})$ which matches the same general exponential form of our lower bound, but
includes an extra constant factor in the exponent.
We conjecture that something like this extra factor may be fundamental in the analysis of this natural strategy. We support this conjecture by noting a straightforward application of a cryptography result due to Pliam~\cite{pliam2000incomparability} indicates that for every constant $\alpha \geq 1$,
there is a random variable $X_{\alpha}$ such that this strategy requires more than $\alpha 2^{H(c(X_\alpha))}$ rounds to succeed with constant probability.\footnote{To be more precise,
the result in~\cite{pliam2000incomparability} can be used to show more than 
$\alpha 2^{H(c(X_{\alpha}))}$
rounds are needed to arrive at the exact correct value. Formalizing this result would require
quantification of the reality that an exact guess is not needed to solve the problem.
This crude application, however, serves its purpose of bolstering our conjecture that $2^{H(c(X))}$
is insufficient.}

An important characteristic of the Kullback–Leibler divergence is that is if each probability
in $Y$ is off by at most a bounded constant fraction from the real probability in $X$,
$D_{KL}(c(X)\Vert c(Y)) = O(1)$.  This establishes that our algorithm does not require a
 precise prediction to be useful, and its efficiency will increase smoothly along with prediction quality.

In Section~\ref{sec:predictions:lower:cd},
we consider algorithms for the collision detector setting.
Given network size variable $Y$, the algorithm first constructs an optimal
code $f$ with respect to source $c(Y)$.
It then considers all codes from shortest to longest in length.
For each length $\ell$, it considers all symbols mapped to codes of this length.
Ordering these symbols from smallest to largest, it deploys the collision detector-driven
binary search strategy introduced by Willard~\cite{willard1986log} to explore
if any of these symbols correspond to the correct network size.

We show this algorithm solves contention resolution
with constant probability in $O((H(c(X)) + D_{KL}(c(X)\Vert c(Y))^2)$ rounds,
which reduces to $O((H(c(X)))^2)$ when the probabilities in $Y$ are all within
a bounded constant factor of the real probabilities in $X$.
As before, our upper bound shares the same general form as the lower bound
(in this case, a result expressed as a polynomial of $H(c(X))$), but is not exactly tight.
Closing these final gaps between the upper and lower bounds for small divergence
is not obvious, but likely tractable; and therefore
left as intriguing future work.
We also note that for clarity, we focused only on one-shot attempts
to resolve contention that succeed with constant probability.
Achieving good bounds on expected time also represents important future work.

\subsection{Contention Resolution with Perfect Advice}

In the interest of further expanding our foundational understanding of contention 
resolution with predictions,
in Section~\ref{sec:perfect} we consider contention resolution with so-called
{\em perfect advice}.
In more detail, we fix parameter $b$ on the maximum number of bits of information
provided to the algorithm as advice or predictions about the network.
We assume an advice function $f_A$ that returns the
same advice $f_A(P)$ of size $b$ bits to every node in the  set $P$ of participating players.
We seek to understand tight bounds on the speed-up possible given $b$ bits of perfect advice.
 Tight bounds here characterize the best possible
improvements possible for a given quantity of advice---lower bounds that can 
later inform the study of more specific prediction models.

\begin{table}
\centering
\renewcommand{\arraystretch}{1.5}
\begin{tabular}{ | c | c | c |}
    \hline
  & Deterministic & Randomized\\ 
 \hline
 No Collision Detection & $\Theta\big(n^{1-b/\log{n}}\big)$ & $\Theta(\log{n}/2^b)$ \\  
 \hline
  Collision Detection & $\Theta(\log{n}-b)$ & $\Theta(\log{\log{n}}-b)$ \\
 \hline
\end{tabular}
\renewcommand{\arraystretch}{1}

\caption{\label{tab:adviceresults} Here we summarize our results for deterministic and randomized contention resolution given $b$ bits
of perfect advice. 
The randomized results are expressed with respect to expected round complexity. All bounds presented are tight.}
\end{table}
We prove this question amenable to analysis by providing tight asymptotic bounds on 
achievable speed-up, with respect to parameter $b$, for both deterministic and randomized algorithms,
with and without collision detection.
Our lower bound results for deterministic algorithms leverage existing lower
bounds on a combinatorial
object called a {\em strongly selective family}~\cite{clementi2003distributed},
whereas our randomized lower bounds leverage a reduction from the
well-understood case where $b=0$.

\subsection{Other Related Work}

There is a long history of studying algorithms for contention resolution in wireless networks, from Chlamtac and Kutten in 1985~\cite{CK85} and Willard~\cite{willard1986log} to the present, including tight upper and lower bounds for the problem in a variety of settings (e.g., with and without collision detection). Several recent directions have influenced this paper.

First, there have been a series of papers recently looking at the minimum number of bits of information per node needed to solve broadcast in a multihop wireless network. In 2006, Fraigniaud, Ilcinkas, and Pelc~\cite{FIP10} proposed measuring the amount of \emph{advice} needed to solve various a distributed problems, and later 
Ilcinkas, Kowalski, and Pelc~\cite{IKP10} applied this to the problem of broadcast in a radio network. Further results on this were given by Ellen, Gorain, Miller, and Pelc~\cite{EGMP19}, Ellen and Gilbert~\cite{EG20}, and Bu, Potop-Butucaru, and Rabie~\cite{BPR20}. These results differed from our results on perfect advice in Section~\ref{sec:perfect} in that they focused on multihop networks, giving a small number of distinct bits of advice to each node. Thus the total number of bits of advice  (e.g., $\Theta(n)$) was significantly larger than in this paper.

Second, there has been much interest lately on how well algorithms can do with additional advice. The general idea is that machine learning and other prediction techniques may be able to \emph{sometimes} provide good advice that allows algorithms to run faster; the challenge lies in ensuring that they continue to perform well when the advice is faulty. Much of this direction began with Mitzenmacher's work on the Learned Bloom Filter in 2018~\cite{M18}, and has continued with a variety of related online algorithms such as work by Kumar, Purohit and Svitkina's~\cite{KPS18} and Lattanzi, Lavastida, Moseley and Vassilvitskii~\cite{LLMV20} on scheduling. Our paper does not focus on machine learning per se, but does look at how extra information can help a distributed algorithm, and examines the impact on performance of faulty advice.

Third, one of the goals of this paper is to construct a connection between information theory and distributed algorithms, and there have been several recent papers that have used information theory and coding theory techniques to better understand radio network algorithms. Dufoulon, Burman, and Beauquier~\cite{DBB20} recently showed how to use techniques from coding theory to communicate efficiently in \emph{beeping} networks; Ashkenazi, Gelles, and Leshem~\cite{AGL20} and Efremenko, Kol, and Saxena~\cite{EKS20} independently showed how to use techniques from coding theory to overcome noise in beeping networks. This series of recent works, along with our own, leads us to believe that there are powerful connections to information theory here.

\section{Contention Resolution with Network Size Predictions}
\label{sec:networksize}

Here we investigate an obvious question. Assume over time you have learned a distribution
that characterizes the likelihood of various network sizes for the instance
of contention resolution you are about to execute. To what degree can you leverage
these predictions to outperform worst-case bounds?
We study this question from both the upper and lower bound perspective,
considering both the collision detection and no collision assumptions.
As elaborated below, in doing so we establish a novel connection
between contention resolution and core information theorey results
regarding coding on noiseless channels.

\subsection{Uniform Algorithms}

We focus on  {\em uniform} contention resolution algorithms, in which participants rely on predetermined probabilities.
With no collision detection, a uniform algorithm can be interpreted as a sequence of probabilities, $p_1,p_2,p_3,\ldots$,
such that in round $i$, all participants broadcast with probability $p_i$.
With collision detection,
a uniform algorithm can be interpreted as a function $f$ from collision histories
to broadcast probabilities. 
For a given round $r$, let the binary string $B = b_1,b_2,\ldots,b_{r-1}$ describe
the collision history through the first $r-1$ rounds of the given execution (i.e., $b_i = 0$ indicates there was no collision in round $i$, while $b_i=1$ indicates there was).
In round $r$, all players transmit with the same uniform probability $f(B)$. 
Uniformity is a common assumption in the study of contention resolution.
Many previous lower bounds assume this property; e.g.,~\cite{willard1986log,jurdzinski2002probabilistic,cover2006}.

\subsection{Preliminaries}
\label{sec:networksize:prelim}
Fix a network size $n$.
Let $X$ be the discrete random variable that determines the number of participants in each instance of contention resolution.
In more detail, $X$ takes its values from $1,2,\ldots,n$, that
occur with probabilities $p_1,p_2,\ldots,p_n$, respectively.
In our setting, the size $k$ of the participant set is determined by $X$, leaving
the adversary only to determine {\em which} $k$ nodes participate.
Notice, however, when considering uniform algorithms the identity of the participants
is not consequential as transmission behavior is determined exclusively by the algorithm and the collision history.

In our algorithms, we sometimes talk about this random variable $X$, or perhaps
an estimate $Y$ of this random variable, being provided {\em as input}
to the algorithm. This is shorthand for the more accurate statement that the
underlying {\em distribution} defining the variable is provided as input.
That is, the algorithm is given for each network size, a prediction
of the probability that the network size is drawn.

In our analysis, we will find it useful to reference a {\em condensed} version of $X$,
we call $c(X)$, which aggregates the probability mass spread over $n$ possible
network sizes into $\log{n}$ geometric size ranges.
(Assume all logs are base $2$.)
To formalize this latter definition, 
let $L(n) = \{1,2,3,\ldots,\lceil \log{n} \rceil\}$.
We associate each range $i\in L(n)$ with the values in the interval $(2^{i-1}, 2^i]$.
That is, $i=1$ is associated with just the value $2$,
$i=2$ is associated with the range $3$ to $4$,
$i=3$ is associated with $5$ to $8$, and so on.
Implicit in this definition is the assumption that the network size is always of size at least $2$ as there is no contention to resolve in a network of size less than $2$.\footnote{This
assumption can hold without loss of generality as all algorithms can eliminate the $n=1$ 
possibility in an additional early round in which all players transmit with probability $1$.}

The random variable $c(X)$ takes its values from $L(n)$. For each $i\in L(N)$, 
let $q_i = \Pr\left(c(X) = i\right)$, defined as follows:

\[q_i = \sum_{2^{i-1} < j \leq 2^i} p_j\]

The analysis that follows  makes use of the following version of Jensen's inequality
applied to concave functions:

\begin{theorem}[Jensen's Inequality]
If $p_1,\ldots,p_n$ are real numbers where $p_i > 0$ for all $i\in[n]$ and $\sum_{i\in[n]}p_i=1$,
and $f$ is a real continuous function that is concave, then:
$f\left(\sum_{i=1}^n p_i x_i\right) \geq \sum_{i=1}^n p_i f(x_i)$.
\label{thm:jensens}
\end{theorem}

It also builds on the lower bound result from Shannon's famed Source Coding Theorem,
which concerns the efficiency of codes on noiseless channels:

\begin{theorem}[Source Code Theorem~\cite{shannon1948mathematical}]
Let $X$ be a random variable taking values in some finite alphabet $\Sigma$.
Let $f$ be a uniquely decodable code from this alphabet to $\{0,1\}$.
Let $S$ be the random variable that describes the length of codeword $f(X)$.
Let $H$ be the entropy function.
It follows: \[ H(X) \leq E(S)\]

\label{thm:shannon}
\end{theorem}

Our upper bound analysis considers the case in which the network size distribution $Y$
provided as input to our algorithms does not exactly match the actual distribution
$X$ from which the network size will be drawn.
Drawing from standard statistics,
we can use the Kullback-Leibler divergence between the two distributions,
denoted $D_{KL}(X\Vert Y)$, to quantify their differences.
We then leverage the following well-known information theory result
which bounds the decrease in coding performance, 
with respect to $D_{KL}(X\Vert Y)$, when you build an construct an optimal
code for $Y$ that you then combine with symbol source $X$ (for more on Kullback-Leibler divergence
and the coding bound see the excellent review in~\cite{cover2006}):

\begin{theorem} 
\label{thm:codelength}
Let $X$ and $Y$ be random variables taking values in some finite alphabet $\Sigma$.
Let $f$ be an optimal, uniquely decodable code from $Y$  to $\{0,1\}$.
Let $S$ be the random variable that describes the length of codeword $f(X)$.
Let $H$ be the entropy function and $D_{KL}$ be the Kullback-Leibler divergence of two distributions.
It follows: \[H(X) + D_{KL}(X\Vert Y)\leq  E(S) \leq H(X) + D_{KL}(X\Vert Y) + 1\]

\end{theorem}

\noindent Lastly, we note note that 
$D_{KL}(X\Vert X)=0$ for any random variable $X$.

\subsection{Lower Bound for No Collision Detection}
\label{sec:predictions:lower:nocd}

Our goal is to prove the following lower bound that connects contention resolution with a known
network size distribution to the entropy of the condensed version of that distribution:

\begin{theorem}
Fix a uniform algorithm $A$ for a network of size $n$.
Let $t_X(n)$ be the expected round complexity for $A$ to solve contention resolution
on a channel with no collision detection and the number of participants determined by random variable $X$.
It follows: $t_X(n) = \Omega\left( \frac{2^{H(c(X))}}{\log\log{n}}\right)$.
\label{thm:lower:nocd}
\end{theorem}

Our proof strategy deploys two steps. We begin by defining a more abstract combinatorial-style 
problem called {\em range finding}, which we can more directly and clearly connect to entropy.
We then show how to transform a contention resolution solution into range finding solution
with a related time complexity.

\paragraph{Range Finding.}
The range finding problem is parameterized
with a network size $n$ and range expressed as a function $f(n)$ of this size.
A range finding strategy can take the form of a sequence of values from $L(n)$,
or a binary tree with its nodes labelled with values from $L(n)$.
Here we define the version defined with respect to a sequence, 
as this is the version needed for our proof of the above theorem.
The binary tree variation will be used when we later consider contention
resolution with collision detection.

We say a sequence $S=v_1,v_2,\ldots,v_k$ {\em solves} the
 $(n, f(n))$-range finding problem in $t$ steps for a given target $v\in L(n)$,
 if $S[t]$ is the first position in $S$ such that $|S[t] - v| \leq f(n)$, where $S[t]=v_t$ is the $t$th element of $S$.
 To handle probabilistic selections of targets from a known distribution,
fix some random variable $Y$ that takes values from $L(n)$.
For each $i\in L(n)$, let $p'_i = \Pr(Y = i)$.
We say $S$ solves $(n,f(n))$-range finding in expected time $t$ with respect to $Y$,
if $t$ is the expected step at which $A$ solves the problem when the target value
is determined by $Y$.

\paragraph{Bounding Sequence Range Finding Using Entropy.}
Assume $S$ is a sequence that solves $(n, f(n))$-range finding in expected time $T$ with respect
to some distribution $Y$ over $L(n)$.
We can use $S$ to design a code for  source $Y$ with an efficiency determined by $f(n)$.
We leverage this connection to prove the following about the connection between
range finding and entropy for a range $f(n) = O(\log\log{n})$ that will prove
useful for our subsequent
attempts to connect contention resolution to range finding:

\begin{lemma}
Let $S$ be a sequence that solves $(n,\alpha\log\log{n})$-range finding for some constant $\alpha\geq 1$ and network size $n>1$.
Assume that the range is determined by random variable $Y$.
Let $Z$ be the random variable describing the complexity of $S$.
It follows: 

\[ E(Z) = \Omega\Big(2^{H(Y)}/(\alpha\log\log{n})\Big), \]

\noindent where $H$ is the entropy function.
\label{lem:code:nocd}
\end{lemma}
\begin{proof}
Fix some $n$, $\alpha$, $S$ and $Y$ as specified by the theorem statement.
We can use $S$ to design a code for transmitting symbols from $L(n)$ over a noiseless channel as follows:

\begin{itemize}
    \item Initialize the sender and receiver with sequence $S$.
    \item To communicate a value $x\in L(n)$,
    the sender transmits $(r,d)$ to the receiver, where: $r$ is the first round
    such that $S[r]$is within $\alpha\log\log{n}$ of $x$ (i.e., the first round
    to solve $(\alpha,n)$-range finding for $x$;
    and $d= x-v_r$ is the distance of $S[r]$ from $x$.
    
    \item The receiver can then locally calculate $x = S[r]+ d$.
\end{itemize}

Let us call the above scheme {\em target-distance} coding.
We now relate the expected code length with this scheme with 
the expected complexity of range finding with $A$.
Consider the performance of the target-distance coding scheme based on these values.
For a given $x\in L(n)$,
that is solved by $S$ by round $r$,
the code length of this scheme is upper bounded by $\log{r} + \log{(\alpha\log\log{n})}$
The expected code length, therefore,
is upper bounded as:

\[ E(\log{Z}) + \lceil \log{(\alpha\log\log{n})}\rceil + 1,\]

where the additional bit indicates the sign of $d$.
Noting the concavity of the log function
we can apply Jensen's Inequality (Theorem~\ref{thm:jensens}) to further refine this bound as follows:

\[ E(\log{Z}) + \lceil \log{(\alpha\log\log{n})}\rceil + 1 \leq \log{E(Z)} + \lceil \log{(\alpha\log\log{n})}\rceil + 1.\]

The Source Code Theorem (Theorem~\ref{thm:shannon}) tells us that the average code length of this scheme
is lower bounded by the entropy of $Y$.
It follows that $H(Y) \leq \log{E(Z)} +\lceil \log{(\alpha\log\log{n})}\rceil + 1$, which implies: 

\begin{eqnarray*}
2^{H(Y)} &\leq& 2^{ \log{E(Z)} + \lceil \log{(\alpha\log\log{n})}\rceil + 1 } \Rightarrow \\
    2^{H(Y)} &\leq & 2^{ \log{E(Z)}}2^{\lceil \log{(\alpha\log\log{n})}\rceil + 1 } \Rightarrow \\
    2^{H(Y)} & \leq & E(Z)\cdot 4 \alpha\log\log{n} \Rightarrow \\
    E(Z) & \geq & \frac{2^{H(Y)}}{4\alpha\log\log{n}}= \Omega\Bigg(\frac{2^{H(Y)}}{\alpha\log\log{n}}\Bigg),
\end{eqnarray*}

\noindent as claimed by the lemma.
\end{proof}

\paragraph{Solving Range Finding with Contention Resolution.}
Here we transform a solution to contention resolution to a sequence that solves
range finding with a similar expected complexity. Contention resolution is a more general
problem than range finding, so care is needed to tame its possible unexpected behaviors.
We begin by defining an algorithmic process for transforming a uniform algorithm $A$
into a range finding sequence $S_A$. We then analyze the properties of $S_A$.

\bigskip

\begin{algorithm}
\SetAlgoLined
\KwIn{Uniform contention resolution algorithm $A=p_1,p_2,\ldots,p_z$}
\KwOut{Range finding sequence $S_A$}
 
 $S_A \gets \emptyset$\;
 $j \gets 0$\;
 
 \For{$i\gets 1$ to $z$}{
    Append $\lceil \log{(1/A[i])} \rceil$ to end of $S_A$\;
    Append $2^j$ to end of $S_A$\;
    $j \gets j + 1$\;
    
    \If{$j > \lceil \log{n}\rceil$}{$j\gets 0$}
 
 }

\Return{$S_A$}
 \caption{RF-Construction}
\end{algorithm}

\bigskip

We now analyze the quality of the range finding solution produced by our
RF-construction algorithm. We begin a useful helper lemma that formalizes
the intuitive notion that a contention resolution algorithm is unlikely
to succeed if its probability is too far form the optimal value
for the participant count.

\begin{lemma}
Assume in a given round of a uniform contention resolution algorithm
that the $1<k\leq n$ participants (where $k\in\mathbb{Z})$ each decide to transmit with 
a probability $p$ such that 
$p < \frac{1}{k\beta\log{n}}$ or $p> \frac{\beta\log{n}}{k}$,
for some sufficiently large constant $\beta \geq 1$.
It follows that the probability exactly one participant transmits
is strictly less than: $\frac{1}{2\log{n}}$.
\label{lem:close}
\end{lemma}

\begin{proof}

If $k$ nodes transmit with probability $p$, the number of nodes which transmit is equal to the binomial distribution $B(k,p)$ and so the probability that a single node transmits is $\Pr(B(k,p)=1)=kp(1-p)^{k-1}$. Therefore, if $p<1/(k\beta\log{n})$ then $\Pr(B(k,p)=1)<(k/(k\beta\log{n}))(1-p)^{k-1}<1/(\beta\log{n})$, satisfying our lemma with any $\beta\geq2$.

\par

Secondly, if instead $p>\beta\log{n}/k$ and $\beta\geq 6$,
\begin{align*}
\Pr(B(k,p)=1)
&< kp \Bigg(1- \frac{\beta\log{n}}{k}\Bigg)^{k-1}\\
&=  kp \Bigg(1- \frac{\beta\log{n}}{k}\Bigg)^{(k/k)(k-1)}\\
&\leq \frac{kp}{e^{(\beta\log{n})(k-1)/k}}\\
&\leq \frac{kp}{e^{(\beta\log{n})/2}}
< \frac{kp}{2^{(\beta\log{n})/2}}\\
&= \frac{kp}{n^{\beta/2}}
<  \frac{k}{n^{\beta/2}}
< 1/n^2
\end{align*}

Where the last line is true for any $\beta\geq 6$ and $2\leq k\leq n$. Therefore, since $1/n^2<1/(2\log{n})$ for all $n\geq 2$ the lemma is satisfied for both cases when $\beta\geq 6$.

\end{proof}

We can now prove our primary lemma which connects the performance of a
contention resolution algorithm to the range finding solution it induces:

\begin{lemma}
Let $A$ be a uniform contention resolution algorithm defined for a network of size $n$
that solves the problem in $t_X(n)$ rounds in expectation when the network size is determined
by random variable $X$.
Let $S_A$ be the range finding sequence returned by RF-construction run on $A$.
There exists a constant $\alpha \geq 1$
such that $S_A$ solves $(n, \alpha\log\log{n})$-range finding in expected time no more than $2t_X(n)$
with respect to $c(X)$.
\label{lem:range:nocd}
\end{lemma}
\begin{proof}
The expected time complexity $t_X(n)$ of $A$ is calculated with respect
to both the probabilistic source of network size $X$ and the bits used by $A$.
Here we introduce the notation required to formalize and manipulate this expectation equation.
Let $\mathcal{S}$ be the sample space of possible random bits generated by the participants running $A$.
We can also enumerate the values in $\mathcal{S}$ as $s_1,s_2,\ldots,s_{z}$.
As a reminder, we have $p_k = \Pr(X = k)$.
Let us also introduce $p'_s$, for $s\in \mathcal{S}$, to be the probability that $s$ are the random bits
generated by participants in $A$.
For participant size $k\in [n]$ and bits $s\in \mathcal{S}$,
let $q(k,s)$ be the number of rounds until $A$ solves contention resolution with $k$ participants
using bits $s$.\footnote{A technicality: this understanding of contention resolution complexity
assumes that the ids of the participants do not matter, only their number. This is clearly the case
for {\em uniform} algorithms of the type studied here, in which participants broadcast according
to a fixed schedule and do not make use of their ids in determining their behavior. The other
technicality is that we do not assume shared.}
And finally, for each such $k$, let $\mu_k$ be the expected round complexity of $A$
when we fix the participant size to $k$.
We note:

\begin{eqnarray*}
t_X(n) & = & \sum_{k=1}^n \sum_{s\in \mathcal{S}} p_k \cdot p_s' \cdot q(k,s) \\
       & = & p_1\left(p_{s_1}'q(1,s_1) + p_{s_2}'q(1,s_2) + \ldots + p_{s_z}'q(1,s_z)\right) + \ldots + \\
       &   & p_n\left(p_{s_1}'q(n,s_1) + p_{s_2}'q(n,s_2) + \ldots + p_{s_z}'q(n,s_z)\right)\\
       & = & p_1 \mu_1 + p_2 \mu_2 + \ldots + p_n \mu_n
\end{eqnarray*}

We now connect the equation in this derivation to our calculation of the expected time complexity of $S_A$ with respect to $c(X)$.
To do so, we first consider a variation of range finding where instead of determining a range with $c(X)$,
we determine a size $k$ with $X$, and then consider the problem solved when $S_A$ arrives at a range sufficiently
close to the corresponding range $\lceil \log{k} \rceil$.
Let $z_k$, for size $k$, be the round where $S_A$ solves range finding for the range corresponding to $k$.
We can compute the expected complexity of this variation as:

\[ p_1z_1 + p_2z_2 + \ldots + p_nz_n\]

Notice, however, that for $i$ and $j$ corresponding to the same range, $z_i = z_j$,
meaning we can aggregate the probabilities associated with range in $L(n)$, 
and get exactly $E(Y)$, defined with respect to $c(X)$.
The above equation, in other words, is an elaborated form of $E(Y)$.
This elaboration is useful because it simplifies our connection of $E(Y)$ to $t_X(n)$,
the expected complexity of $A$, calculated above.

In particular, we will next argue that for each $i$, $z_i \leq 2\mu_i$.
To do so, we consider two cases for a given $\mu_i$:
\begin{itemize}
    \item {\em Case 1: $\mu_i \leq \log{n}$:} Here we deploy a key operational property
    of $A$: if its expected complexity for a given participant size is small, it must feature a good probability
    for that participant size early on.
    Formally, assume for contradiction that the first $\mu_i$ probabilities in $A$ fall
    outside the range from $1/(\beta\cdot \log{n}\cdot i)$ to $(\beta\log{n})/i$ specified by Lemma~\ref{lem:close}.
    By this lemma, the probability of success in each of these rounds is therefore strictly less than $1/(2\log{n})$.
    Because $\mu_i \leq \log{n}$,
    a union bound establishes that the probability that at least one of these rounds succeeds is strictly less than $1/2$.
    It would follow that the probability that $A$ succeeds in the first $\mu_i$ rounds with network size $i$ is less than
    $1/2$, contradicting the assumption that $\mu_i$ is the expected time complexity in this context.
    
    It follows then that there is a probability $p^*$ between the values of $1/(\beta\cdot \log{n} \cdot i)$ to $(\beta\log{n})/i$
    in the first $\mu_i$ rounds of $A$.
    In RF-construction, this probability becomes guess $x =\lceil \log{(1/p^*)} \rceil$.
    Therefore:
    
    \begin{eqnarray*}
    x &\geq & \lceil \log{(i/(\beta\log{n}))} \rceil\\
     &= & \lceil \log{i} - \log{(\beta\log{n})} \rceil\\
     &= & \lceil \log{i} - (\log\log{n} + \log{\beta}) \rceil\\
     &>& \lceil \log{i} - \log{\beta}\log\log{n} \rceil.
    \end{eqnarray*}

   Bounding the other direction, and apply a similar derivation, we get:
   \[ x \leq \lceil \log{(\beta \cdot \log{n} \cdot i)} \rceil \leq \lceil \log{i} + \log{\beta}\log\log{n} \rceil.\]
    
   It follows that this for $\alpha \geq \log{(\beta)}$, the slot in $S_A$ corresponding to $p^*$ solves the $(n, \alpha\log\log{n})$-range finding.
   Because the construction of $S_A$ interleaves values between those corresponding to the probabilities in $A$,
   the position of $p^*$ in $S_A$ could be up to a factor of $2$ larger than its position in $A$.
   Therefore, it shows up within the first $2\mu_i$ positions in $S_A$, satisfying our claim.
   
   \item {\em Case 2: $\mu_i > \log{n}$:} This is the easier case. By the construction of $S_A$,
   during the first $2\log{n}$ rounds we interleave values corresponding to all ranges.
   Therefore, by definition, $z_i \leq 2\log{n} \leq 2\mu_i$, as needed.
   
\end{itemize}

Pulling together the pieces, we have  established: 

\begin{eqnarray*}
E(Y) & = & p_1z_1 + p_2z_2 + \ldots  +p_nz_n\\
    & \leq & p_1(2\mu_1) + p_2(2\mu_2) + \ldots + p_n(2\mu_n)\\
    & = & 2\left( p_1\mu_1 + p_2\mu_2 + \ldots + p_n\mu_n  \right)\\
    & = & 2t_X(n),
\end{eqnarray*}

 as claimed by the lemma.
\end{proof}

\paragraph{Pulling Together the Pieces.}
Fix some uniform contention resolution algorithm $A$ that solves contention resolution in expected time $t_X(n)$
when run in a network of size $n$ with no collision detection and a participant size determined by $X$.
By Lemma~\ref{lem:range:nocd},
there exists a constant $\alpha \geq 1$,
such that the range finding sequence $S_A$ constructed by applying RF-construction on $A$,
solves $(n, \alpha \log\log{n})$-ranging finding in expected time $T \leq 2t_X(n)$ when ranges are drawn from $c(X)$.

Applying Lemma~\ref{lem:code:nocd} further tells us $T \geq 2^{H(c(X))}/(\alpha\log\log{n})$.
It then follows that $2t_X(n) \geq 2^{H(c(X))}/(\alpha\log\log{n}) \Rightarrow t_X(n) =$ $\Omega\left( \frac{2^{H(c(X))}}{\log\log{n}} \right)$,
which proves Theorem~\ref{thm:lower:nocd}.

\subsection{Lower Bound for Collision Detection}
\label{sec:predictions:lower:cd}

We now adapt the techniques used in the preceding section to achieve an entropy-based
lower bound for the setting with collision detection.
Our goal is to prove the following:

\begin{theorem}
Fix a uniform algorithm $A$ for a network of size $n$.
Let $t_X(n)$ be the expected round complexity for $A$ to solve contention resolution
on a channel with collision detection and the number of participants determined by random variable $X$.
It follows: $t_X(n) \geq  (1/2)H(c(X)) - O(\log\log\log\log{n})$.
\label{thm:lower:cd}
\end{theorem}

For the maximum possible entropy value $H(c(X)) = \log\log{n}$, 
this bound asymptotically matches the best known upper bound, from the 1986 work of Willard~\cite{willard1986log},
which requires $O(\log\log{n})$ rounds.
It also provides an arguably simpler and more intuitive 
approach than the original lower bound from~\cite{willard1986log},
which deployed a more complex probabilistic counting argument to establish $\log\log{n} - O(1)$ rounds
as being necessary for uniform algorithms.

We note that the appearance of a quadruple logarithm is unusual, but straightforward to explain
in our context.
In the argument that follows we seek a probability within a factor of $1/\log\log{n}$ from the optimal
probability for the current participant size. This is a $\log$ factor closer than in our argument
for no collision detection, as the shorter executions can handle smaller error probabilities.
Recall, within our condensed support $L(n)$, 
each $i\in L(n)$ is associated with probability $2^{-i}$.
So if $2^{-i}$ is the optimal probability, than a range $j$ that is within a distance of $\log\log\log{n}$
from $i$ will yield a probability within a factor of $1/\log\log{n}$; e.g., 

\[2^{-j} = 2^{-(i+\log\log\log{n})} = 1/(2^i2^{\log\log\log{n}}))= 1/(2^i\log\log{n}).\] 

In the coding scheme used below, as in the no collision detection argument, 
a code contains a  distance value from $0$ to $O(\log\log\log{n})$, which
requires $O(\log\log\log\log{n})$ bits.

Our proof below follows the same structure as in the no collision detection case.
We differ, however, in the  details of how we construct our range finding solution and bound such solutions
form an information theory perspective.

\paragraph{Bounding Tree Range Finding Using Entropy.}
Assume $T$ is a binary tree that solves $(n,\alpha\log\log\log{n})$-range finding in expected time $T$ with respect
to some distribution $Y$ defined over $L(n)$.
By interpreting $T$ as a code,
we can deploy Shannon's source coding theorem to arrive at the following bounds:

\begin{lemma}
For some constant $\alpha\geq 1$ and network size $n>1$, let $T$ be a labeled binary tree that solves $(n,\alpha\log\log\log{n})$-range finding.
Assume the target range be determined by random variable $Y$.
Let $Z$ be the random variable describing the complexity of solving range finding using $T$ with respect to $Y$.
It follows: 

\[ E(Z) \geq H(Y) - O(\log\log\log\log{n}), \]

\noindent where $H$ is the entropy function.
\label{lem:code:cd}
\end{lemma}
\begin{proof}
Fix some $n$, $\alpha$, $T$ and $Y$ as specified by the theorem statement.
We can use $T$ to design a code for transmitting symbols from $L(n)$ over a noiseless channel as follows:

\begin{itemize}
    \item Initialize the sender and receiver with tree $T$.
    
    \item To communicate a value $x\in L(n)$,
    the sender transmits $(p,d)$ to the receiver, where: 
    $p$ is a binary sequence describing a path from the root to the lowest-depth occurrence in the tree of a value $v$
    within distance $\alpha\log\log\log{n}$ of $x$ (a $0$ bit means extend the path by descending to the left
    sub-tree and a $1$ bit means descend to the right); and $d = |v-x|$, the distance from $v$ to $x$.

    \item The receiver can then locally calculate $x$ by traversing to the node labeled $v$ in $T$ then adding $d$.

\end{itemize}

We now relate the expected code length with this scheme with 
the expected complexity of range finding.
For a given $x\in L(n)$,
that is solved by $T$ by a node with value $v$ at depth $h$,
the code length of this scheme is upper bounded $h + \lceil \log{(\alpha\log\log\log{n})} \rceil$,
where the quadruple $\log$ factor encodes the distance of $v$ from $x$.
The expected code length, therefore,
is upper bounded as:

\[ E(Z) + \lceil \log{(\alpha\log\log\log{n})}\rceil.\]

The Source Code Theorem (Theorem~\ref{thm:shannon}) tells us that the average code length of this scheme
is lower bounded by the entropy of $Y$.
It follows:

\begin{align*}
    &H(Y) \leq E(Z) + \lceil \log{(\alpha\log\log\log{n})} \rceil \\ 
    & \Rightarrow E(Z) \geq H(Y) - O(\log\log\log\log{n}),
\end{align*} 
\noindent as claimed by the theorem.
\end{proof}

\paragraph{Solving Range Finding with Contention Resolution.}
Here we transform a solution to contention resolution with collision detection
to a binary tree that solves
range finding with a related expected complexity. 
We begin by defining an algorithmic process for transforming a uniform algorithm $A$
into a range finding tree $T_A$. 

As a uniform contention resolution algorithm that assumes collision detection,
$A$ is formalized as a function that maps the sequence of collisions and silences detected
so far into a probability for all participants to use during the next round.
A history of length $r$ can be captured by a bit sequence $b_1b_2\ldots b_r$,
where $b_i = 0$ means silence was detected in round $i$ and $b_i = 1$ means a collision was detected.

We can interpret $A$ as a binary tree where each node is labeled with a probability.
In particular, interpret each input string $s=b_1b_2\ldots b_r$ as specifying 
a particular node at depth $r-1$, reached in a $r$-step traversal starting
from the root, where at step $i$ you descend the to left sub-tree if $b_i = 0$
and descend to the right sub-tree if $b_i = 1$.
You labeled this node with probability $A(s)$; that is, the probability mapped
to pattern $s$ by the algorithm.

Let $T_1$ be this binary tree labeled with probabilities.
We next create tree $T_2$ by replacing each label $\ell$ in $T_1$
with its related range: $\lceil \log{(1/\ell)} \rceil$, as in our no collision detection construction.
Let $T^*$ be the canonical binary tree of depth $\lceil \log\log{n} \rceil$ labeled with all
the values in $L(n)$.
To arrive at our final range finding solution $T_A$,
we must insert $T^*$ into the tree, with the root beginning at depth $\lceil \log\log{n} \rceil$.
There are many equally useful ways to do so.
Assume for now that we just follow the left-most path through $T_2$,
and when arrive at node $v$ depth $\lceil \log\log{n} \rceil$,
we remove $v$'s children and instead make the root of $T^*$ the only child of $v$.

We now analyze the quality of the range finding tree $T_A$ produced by our above procedure.
We begin with a useful probability observation.

\begin{lemma}
Assume in a given round of a uniform contention resolution algorithm
that the $1<k\leq n$ (where $k\in\mathbb{Z})$ participants each decide to transmit with 
a probability $p$ such that 
$p < \frac{1}{\beta(\log{\log{n}})k}$ or $p > \frac{\beta(\log\log{n})}{k}$,
for some sufficiently large constant $\beta \geq 1$.
It follows that the probability exactly one participant transmits
is strictly less than: $\frac{1}{2\log\log{n}}$.
\label{lem:close:cd}
\end{lemma}

\begin{proof}

In the case where $p<1/(k\beta\log{\log{n}})$, a simple union bound over all $k$ nodes yields that the probability any node transmits is strictly less than $k/(k\beta\log{\log{n}})=1/(\beta\log{\log{n}})$. This satisfies our lemma as long as $\beta\geq 2$.

\par

Again recall that the number of nodes which transmit is given by the binomial distribution $B(k,p)$. Therefore for the case that $p>(\beta\log{\log{n}})/k$ we use the Hoeffding's inequality applied to the binomial distribution, $\Pr(B(k,p)\leq k')\leq \exp(-2k(p-k'/k)^2)$.

\begin{align}
    \label{line:subp}
    \Pr(B(k,p) = 1)
    &< \Pr(B(k,p)\leq 1)
    \leq \frac{1}{ \exp(2k(p-1/k)^2)}\\
    &<\frac{1}{ \exp(2k((\beta\log{\log{n}})/k-1/k)^2)}\\
    \label{line:subk}
    &=\frac{1}{ \exp((2/k)(\beta\log{\log{n}}-1)^2)}\\
    &\leq \frac{1}{\exp((\beta\log{\log{n}}-1)^2)} <\frac{1}{2\log{\log{n}}}
\end{align}

Where Line \ref{line:subp} holds due to our assumption $p>(\beta\log{\log{n}})/k$ and Line \ref{line:subk} holds for $k\geq 2$ and $\exp((\beta\log{\log{n}}-1)^2)>2\log{\log{n}}$ (which if $\beta \geq 2$ is true for all $n\geq 4$). We therefore have that $\beta\geq 2$ satisfies the lemma in both cases.





\end{proof}

We now make our main argument analyzing $T_A$'s performance as a solution to range-finding.

\begin{lemma}
Let $A$ be a uniform contention resolution algorithm defined for a network of size $n$ with collision detection,
that solves the problems in $t_X(n)$ rounds in expectation when the network size is determined
by random variable $X$
Let $T_A$ be the range finding tree returned by apply our above procedure to  $A$.
There exists a constant $\alpha \geq 1$
such that $T_A$ solves $(n, \alpha\log\log\log{n})$-range finding in expected time no more than $2t_X(n)$
with respect to $c(X)$.
\label{lem:range:cd}
\end{lemma}
\begin{proof}
Redeploying the same argument as in the proof Lemma \ref{lem:range:nocd},
we can express $t_X(n)$ in the following useful form:

\[ p_1 \mu_1 + p_2 \mu_2 + \ldots + p_n \mu_n, \]

\noindent where $p_i = \Pr(X=i)$ and $\mu_i$ is the expected round complexity of $A$ given
that $i$ is the network size.
We use this form of $t_X(n)$ to help bound the expected complexity of $T_A$ with respect to $c(X)$.
Also as in the proof of Lemma~\ref{lem:range:nocd}, 
we can consider a variation of range finding where instead of directly determining a range with $c(X)$,
we determine a size $k$ with $X$, and then transform this into the corresponding range $\lceil \log{k}\rceil$.
As argued earlier, 
calculating the expected complexity of $T_A$ for this variation
is mathematically equivalent to the complexity for the standard variation in which
values are drawn from $c(X)$.

We  compute the expected complexity of $T_A$ running this variation as:

\[ p_1z_1 + p_2z_2 + \ldots + p_nz_n,\]

\noindent where $z_i$ is the depth of the first value in tree $T_A$ to fall
within a sufficient distance of the range $\lceil \log{i} \rceil$ corresponding to size $i$.
We will next argue that for each $i$, $z_i \leq 2\mu_i$.

To do so, we consider two cases for a given $\mu_i$:
\begin{itemize}
    \item {\em Case 1: $\mu_i \leq \log\log{n}$:} Consider $T_1$,
    our interpretation defined earlier
    of $A$ as a binary tree labeled with $A$'s broadcast probabilities.
    We argue here that within depth $\mu_i$ of $T_1$
    we can find a probability within a range of $\beta\log\log{n}$ of $1/i$,
    where the constant $\beta$ comes from Lemma~\ref{lem:close:cd}.
    
    To make this argument, assume for contradiction that {\em no} probability
    within depth $\mu_i$ was within this range.
    By Lemma \ref{lem:close:cd},
    given any execution of $A$,
    during any of the first $\mu_i \leq \log\log{n}$ rounds of the execution,
    the probability of solving contention resolution in that round is strictly
    less than $1/(2\log\log{n})$. By a union bound, the probability
    of succeeding in the first $\mu_i$ rounds is itself strictly
    less than $1/2$.
    This contradicts the assumption that the expected round in which
    contention resolution is solved is $\mu_i$.
    
    We have established, therefore,
    that somewhere in the first $\mu_i$ levels of the tree
    is a probability $p^*$in the range $1/(\beta\cdot (\log\log{n}) \cdot i)$ to $(\beta(\log\log{n}))/i$.
    When constructing $T_A$,
    this value will be replaced by $x=\lceil \log{(1/p^*)} \rceil$,
    which for a sufficiently large constant $\alpha$ (defined with respect to $\beta$),
    is within distance $\alpha\log\log\log{n}$ of the target  $\lceil log(1/i) \rceil$,
    solving range finding for this value.
    It follows that $z_i \leq \mu_i$.

   \item {\em Case 2: $\mu_i > \log\log{n}$:} This is the easier case. By the construction of $T_A$,
   every rage shows up somewhere between depth $\log\log{n}$ and $2\log\log{n}$,
   as we inserted tree $T^*$ into $T_A$ starting at that depth.
   It follows that the depth at which the problem is solved for size $i$ is less than $2\mu_i$ (i.e., if
   $\mu_i=\log\log{n}+1$ and the corresponding range is at the lowest level of $T^*$).

\end{itemize}

Combined, these two cases establish the for a sufficiently large constant $\alpha$,
the expected complexity of $(n,\alpha\log\log\log{n})$-range finding with $T_A$ and $c(X)$
is with a factor of $2$ of the expected complexity of $A$ with $X$, as claimed by the lemma.

\end{proof}

\paragraph{Pulling Together the Pieces.}
Fix some uniform contention resolution algorithm $A$ that solves contention resolution in expected time $t_X(n)$
when run in a network of size $n$ with collision detection and a participant size determined by $X$.
By Lemma~\ref{lem:range:cd},
there exists a constant $\alpha \geq 1$,
such that the range finding tree $T_A$ constructed by applying our procedure to $A$,
solves $(n, \alpha\log\log\log{n})$-ranging finding in expected time $T \leq 2t_X(n)$ when ranges are drawn from $c(X)$.
Applying Lemma~\ref{lem:code:cd} further tells us that $T \geq H(c(X)) - O(\log\log\log\log{n})$.
It follows that $2t_X(n) \geq H(c(X)) - O(\log\log\log\log{n}) \Rightarrow t_X(n) = (1/2)H(c(X)) - O(\log\log\log\log{n})$,
which proves Theorem~\ref{thm:lower:cd}.

\subsection{Upper Bound for No Collision Detection}
\label{sec:predictions:upper:nocd}

In this section we introduce an algorithm for solving contention resolution without collision detection.
We assume the algorithm is provided as input the definition of a random
variable $Y$ defined over network sizes.
Let $X$ be the actual random variable from which the sizes will be drawn.
We will produce our round complexity bounds with respect to the statistical
divergence between $Y$ and $X$, quantifying the cost of inaccuracy in predictions.

In Section~\ref{sec:predictions:lower:nocd}, we proved
that with accurate predictions (i.e., $Y=X$),
$\Omega(2^{H(c(Y))}/\log\log{n})$ rounds are needed in expectation
to solve contention resolution.
Our goal here is produce a result that comes close to matching this 
exponential bound.
To do so, we analyze a natural strategy:  trying range predictions
in $c(Y)$ in decreasing order of likelihood.
We prove that with constant probability, this strategy
solves the problem in $2^T$ rounds,
where $T = 2H(c(X)) + 2D_{KL}(c(X)\Vert c(Y))$.
If every probability used by $Y$ is within some bounded constant factor
of the corresponding probability in $X$,
this reduces to $2^{2H(c(Y))}$,
which is within the same general form as the lower bound but with an extra
factor in the exponent.
It is not obvious how to remove any extra factor in the exponent.
Indeed, as elaborated in Section~\ref{sec:introduction},
we have reason to believe that some exponential factor greater than $1$
is necessary for this algorithm.

\subsubsection {Algorithm}

Let $\pi=\langle \pi_1,\ldots,\pi_{\log{n}}\rangle $ represent an ordering over $L(n)$ sorted by non-decreasing probability of the corresponding range with respect to $c(Y)$. In other words such that for all  $i<j$, $\Pr(k \in (2^{\pi_i-1}, 2^{\pi_i}] \geq \Pr(k \in (2^{\pi_j-1}, 2^{\pi_j}]$. Our algorithm consists of $\log{n}$ rounds where in round $i$ each node broadcasts with probability $1/2^{\pi_i}$.\footnote{Because
we seek only a constant probability result, we analyze here only
one pass through all $\log{n}$ possible probabilities.
In the pursuit of good expected times, you would instead cycle through
these probabilities in a clever manner. We do not prove expectation
bounds on this algorithm here, and note that an expectation close
to our constant probability result is not necessarily easily obtained.}

\subsubsection{Analysis}

We prove the following time complexity statement regarding the above algorithm.

\begin{theorem}
\label{thm:constprob}

In the above algorithm, a node broadcasts alone after at most $O(2^T)$ rounds where $T=2H(c(X)) + 2D_{KL}(c(X)\Vert c(Y))$ rounds with probability at least $1/16$, where $X$ is the actual distribution over the network sizes and $Y$ is the distribution provided to the algorithm.

\end{theorem}

First we prove that in round $i$ that if $k \in (2^{\pi_i-1}, 2^{\pi_i}]$ then a single node broadcasts alone during the round with probability at least $1/2$.

\begin{lemma}
\label{lem:testprob}

For $k\geq 2$ and $i>0$ if $k \in (2^{\pi_i-1}, 2^{\pi_i}]$, a single node broadcasts alone in round $i$ with probability at least $1/8$.

\end{lemma}

\begin{proof}

If $k$ nodes broadcast with probability $p$, the probability that a single node broadcasts alone is equal to $\Pr[B(k,p)=1]$ where $B$ denotes the binomial distribution. In phase $i$, nodes broadcast with probability $p=1/2^{\pi_i}$ which by our assumption then means $p\in (1/(2k), 1/k]$. Therefore, the probability a single node broadcasts in a given round of phase $i$

\begin{align}
    \Pr(B(k,p)=1)
    &=\binom{k}{1}p(1-p)^{k-1}\\
    \label{line:range}
    &> k (1/(2k)) (1- 1/k)^{k-1}
    =(1/2)(1- 1/k)^{k-1}\\
    \label{line:prob}
    &\geq (1/2)(1/4)^{(1/k)(k-1)}\geq (1/2)(1/4)=1/8
\end{align}

Note that Line \ref{line:range} holds for $p\in [1/(2k), 1/k)$ and Line \ref{line:prob} uses the inequality $1-p\geq (1/4)^p$ for $p\leq 1/2$ which is true for any $k\geq 2$.

\end{proof}

Consider an optimal variable-length code $f$ over the values of $L(n)$ based on their probabilities according to $c(Y)$.
We will relate this code to the number of rounds required by our algorithm.

\begin{lemma}
\label{lem:sizetorounds}

With at least probability $1/8$ the above algorithm takes $2^{S + 1}$ rounds where $S$ is the random variable representing the code lengths of $f$ applied to $c(Y)$.

\end{lemma}

\begin{proof}

From Lemma \ref{lem:testprob} we know that in round $i$ of our algorithm where $k \in (2^{\pi_i-1}, 2^{\pi_i}]$ we succeed with probability at least $1/8$. Fix this round $i$ and let $f(\pi_i)$ be the codeword assigned to $\pi_i$ under $f$ according to $c(Y)$. Note that since we assume $f$ is optimal we have that $|f(\pi_i)|\geq |f(\pi_j)|$ for all $j<i$. Therefore, with $S$ as the random variable over this codeword length we have that the total number of estimates preceding $\pi_i$ is at most $2^{S+1}-1$ and the number of rounds for our algorithm to test a value that yields constant success probability is given by $2^{S+1}$.

\end{proof}

\begin{proof}(of Theorem \ref{thm:constprob}.) From Lemma \ref{lem:sizetorounds} we have that the distribution over the number of rounds of our algorithm is given by $2^{S}+1$. Furthermore from Theorem \ref{thm:codelength} we have that for an optimal code $f$ over distribution $c(Y)$ where $c(X)$ is the actual distribution, $\mathbf{E}(S) \leq H(c(X)) + D_{KL}(c(X)\Vert c(Y)) + 1$. Therefore, applying Markov's inequality gives us that the probability of the code length corresponding to the correct estimate is at most $2(H(c(X)) + D_{KL}(c(X)\Vert c(Y)) + 1)$ with probability at least $1/2$. Multiplying the probability of this event with the success probability from \ref{lem:testprob} then satisfies our theorem statement.

\end{proof}

Note that since $D_{KL}(c(X)\Vert c(X))=0$, if instead our algorithm learns from the actual distribution $c(X)$ we can bound the running time with respect to the entropy of $c(X)$ alone.

\begin{corollary}(of Theorem \ref{thm:constprob}.) \label{cor:constprob} In the above algorithm, a single node broadcasts alone after at most $O(2^{2H(c(X))})$ rounds with probability at least $1/16$, where $X$ is both the actual distribution over the network sizes and the distribution provided to the algorithm.

\end{corollary}

\subsection{Upper Bound for Collision Detection}
\label{sec:predictions:upper:cd}

We now turn our attention to a setting with collision detection.
We once again make use of an optimal code $f$ 
constructed for source $c(Y)$,
and use its code words to structure our algorithm's behavior.

\paragraph{Algorithm}

 We group the values from $L(n)$  into $x$ equivalence classes based on the length of their code according to $f(c(Y))$. Let $\pi_1,\ldots, \pi_{x}$ be these classes where the $i$th class contains all values from $c(Y)$ which have codes of length exactly $i$. More formally, for all $i\in[x]$, let $\pi_i = \{j\in L(n)\mid |f(j)|=i\}$. Note that $x\leq \log{\log{n}}$ since we can assign a unique code of length $\log{\log{n}}$ bits to all ranges, so the existence of larger codes
 would contradict the assumed optimality of $f$.

Our algorithm then divides rounds
into $x$ phases,
one dedicated to each $\pi_i$.
In the
 phase for $\pi_i$,
 we use transmissions and collision
detection to perform a binary search over the possible network size ranges represented by the values in class $\pi_i$. The binary search algorithm we use is an adaptation of the classical strategy presented in \cite{willard1986log} which searches over 
a collection of $\lceil \log{n} \rceil$ geometrically distributed 
network size guesses, transmitting with a corresponding probability
for each guess, and using collision and silence
to indicate if a guess is too small or too large, respectively.

 More formally, when searching over $\pi_i$, the nodes order the ranges in $\pi_i$ from smallest to largest. They then broadcast with probability $2^{-m}$, where $m$ is the median of these values. If a collision is detected, the nodes then recurse over the values greater than $m$. Otherwise, if silence is detected, they recurse over the values smaller than $m$. If a single node broadcasts, then contention resolution is solved.
 
 The algorithm proceeds through the phases in order
 of the classes; i.e., $\pi_1$ then $\pi_2$, and so on.
 If the problem is not solved during the search
 for class $\pi_i$, then it moves on to search $\pi_{i+1}$.
 As without our no collision detection algorithm,
 we present this result here as a one-shot attempt
 that solves contention resolution with constant probability.
 For higher probability, it can be repeated, but we 
 do not analyze this form.
 

\par

\paragraph{Analysis}

Our goal is to prove the following about the time complexity of our algorithm.

\begin{theorem}
\label{thm:cdconstprob}
In the above algorithm, 
with constant probability:
contention resolution is solved in  $O((H(c(X)) + D_{KL}(c(X)\Vert c(Y))^2)$ rounds, 
where $X$ corresponds to the actual distribution over the network sizes and $Y$ is the distribution provided to the algorithm.
\end{theorem}

We first analyze our algorithm with respect to the random variable describing the code lengths generated by $f$ when symbols are drawn from $c(Y)$.

\begin{lemma}
\label{lem:cdlength}
With constant probability: the algorithm 
solves contention resolution in  $O(S^2)$ rounds where $S$ is the random variable describing the code lengths of $f$ applied to $c(Y)$.
\end{lemma}

\begin{proof}
The analysis of Willard's strategy from~\cite{willard1986log}
provides the following useful property of this search strategy:
if class $\pi_j$ contains the target value $t = \lceil \log{k} \rceil$,
where $k$ is the actual network size,
then with constant probability the search of this class solves contention resolution.
Note, by definition, all code words corresponding to values in $\pi_j$ are of length $j$.
Therefore, $|\pi_j| \leq 2^j$, meaning
the search for phase $j$ requires no more than $O(\log{|\pi_j|}) = O(j)$ rounds.

Therefore, if the target is in $\pi_j$, which is equivalent to saying
the symbol has a code of length $j$,
then we will complete the search for phase $j$ after first completing searches
for phases $1,2,\ldots,j-1$. Total, this requires $O(j^2)=O(S^2)$ rounds.
Therefore, we solve the problem with constant probability after $O(S^2)$ rounds, as claimed.
\end{proof}

We now prove our final result by leveraging the optimality of the code
to express the expected code length with respect to the entropy of $c(Y)$.

\begin{proof}(of Theorem \ref{thm:cdconstprob}) We have from Lemma \ref{lem:cdlength} that 
the algorithm solves contention resolution with constant probability in $O(S^2)$ rounds, 
with constant probability, where $S$ is the random variable describing the length
of the code word for the target value drawn from $c(Y)$.
 We have from our optimal code $f$ and Theorem \ref{thm:codelength} that the expected value of $S$ is at most $H(c(X)) +D_{KL}(c(X)\Vert c(Y)) + 1$. Therefore from Markov's inequality we have that $S\leq 2(H(c(X)) +D_{KL}(c(X)\Vert c(Y)) + 1)$ with probability at least $1/2$. 
 
To obtain our final result, we multiply the probability that $S$ is not larger than the
bound from the theorem with the probability that we solve contention resolution in $O(S^2)$
rounds. Both values are constant, providing the claimed constant success probability.

\end{proof}

As with Theorem \ref{thm:constprob}, because $D_{KL}(c(X)\Vert c(X))=0$, if $Y=X$,
then we get a sharper bound:

\begin{corollary}(of Theorem \ref{thm:cdconstprob}.) \label{cor:cdconstprob} 
In the above algorithm, if the input distribution $Y$ equals the actual network size distribution $X$:  a single node broadcasts alone after at most $O(H^2(c(X)))$ rounds.
\end{corollary}

\section{Contention Resolution with Perfect Advice}
\label{sec:perfect}

In the previous section,
we studied how to speed up contention resolution given
probabilistic predictions on the network size.
We proved bounds with respect to the {\em quality} of the predictions.
Here we investigate bounds on speed up proved with respect to
the {\em size} of predictions.
More concretely, we assume a general setting,
in which some abstract learning model provides $b$ bits of predictive
advice to the algorithm.
Our goal is to understand for a given advice size $b$,
the theoretical limit on the speed up possible given this much information.
That is, given the best possible advice of this size,
how much improvement is possible?

For $b=0$, the worst case lower bounds apply.
For $b\geq \log{n}$, it is possible to encode the id of a participant in $b$ bits,
enabling a solution in only $1$ round.
We study the space between these extremes,
providing tight bounds for both deterministic and randomized algorithms.
The goal with these bounds is to set a ceiling on how much improvement we can
ever expect from a given specific learning model with bounded advice size.

\subsection{The Perfect Advice Model}

We study the {\em perfect advice model}, a generalization
of the standard contention resolution model,
parameterized with an advice size $b\geq 0$,
in which we augment a contention resolution algorithm $A$
with an {\em advice function}, $f_A: \mathbb{P}(A) \rightarrow \{0,1\}^b$.
At the beginning of each execution,
after the adversary selects the set $P \subseteq V$ of players to participate,
each player in $P$ is provided the same $b$ bits of advice $f_A(P)$.
The advice function, in other words, has perfect knowledge of the participants
for the current execution, allowing it generate the best possible prediction
for the given size bound.

\subsection{Tight Bounds for Deterministic Algorithms}
An interesting place to start this investigation of perfect advice
is with deterministic algorithms, as these allow us to study the impact
of $b$ bits of advice from a purely combinatorial perspective; e.g., asking how many
relevant participant sets can be eliminated given a specific advice string?

With no advice (i.e., $b=0$), deterministic
collision detection requires $\Theta(n)$ rounds without collision detection
and $\Theta(\log{n})$ rounds with collision detection~\cite{newport2014radio}.
We prove that with $b=\alpha\log{n}$ bits of advice,
for positive  $\alpha <1$,
the problem still requires at least $n^{1-\alpha}/2$ rounds
 to solve without collision detection.
With collision detection, $b$ bits of advice can improve the bound at best
to $\log{n} - b$ rounds.
We show both bounds to be tight within small constant factors.

Our strategy is to first bound a harder problem that we call {\em non-interactive contention resolution} that forces nodes to decide whether or not to transmit in a single round, based only on their advice. Leveraging existing theory on a combinatorial object known as a 
strongly selective family~\cite{clementi2003distributed}, we formalize the intuitive result that very close to $\log{n}$ bits of advice are needed to solve this problem deterministically. That is, there is no strategy much better than simply having
the advice function specify exactly which single participant should transmit. 
We  then leverage this result as a foundation for our lower bounds on the round complexity of normal contention resolution with and without collision detection.
These bounds use efficient contention resolution algorithms to construct
efficient solutions to non-interactive resolution.

\paragraph{Non-Interactive Contention Resolution.}
We say an algorithm $A$ and advice function $f_A$ solves {\em $b(n)$-non-interactive contention resolution},
if for every participant set $P \subseteq V$,
the advice $f_A(P)$ contains no more than $b(n)$ bits and leads to exactly one participant from $P$ transmitting in the first round.

Our goal is to show that this requires $b(n) = \Omega(\log{n})$.
To do so, we leverage bounds on the following object:

\begin{definition}
Fix integers $n,k$, $1 \leq k \leq n$.
A family $\mathbb{F}$ of subsets if $[n]$
is {\em $(n,k)$-strongly selective} if
for every subset $Z$ of $[n]$ such that $|Z| \leq k$ and for every element $z\in Z$, there is a set $F$ in $\mathbb{F}$
such that $Z \cap F$ = $\{z\}$.
\end{definition}

\noindent As proved in~\cite{clementi2003distributed},
for large $k$, there are no small strongly selective families:

\begin{theorem}[from~\cite{clementi2003distributed}]
Let $\mathbb{F}$ be an $(n,k)$-strongly selective family.
If $k \geq \sqrt{2n}$ then it holds that $|\mathbb{F}| \geq n$.
\label{thm:ssf}
\end{theorem}

\noindent We can leverage this fact to establish our lower bound on non-interactive
contention resolution.

\begin{theorem}
Assume deterministic algorithm $A$ and advice function $f_A$ solve $b(n)$-non-interactive contention resolution in a network of size $n$.
It follows that $b(n) \geq \log{n}$.
\label{thm:noninteractive}
\end{theorem}
\begin{proof}
Let $S$ contain all $2^{b(n)}$ possible advice strings.
For each $s\in S$,
let $V(s)$ be the nodes in $P$ that would broadcast
according to $A$ given advice $s$.
By the assumption that $A$ is correct,
for each $P\subseteq V$,
with corresponding advice $s_P = f_A(P)$: $|V(s_P) \cap P| = 1$.
We can therefore consider $\mathbb{F} = \{ V(s) \mid s \in S\}$
to be an $(n,n)$-strongly selective family.
It follows from Theorem~\ref{thm:ssf} 
that $|\mathbb{F}| \geq n$,
which implies $|S| \geq n$,
which implies the number of bits required
to encode elements of $S$ is at least $\log{n}$, as claimed by the theorem.
\end{proof}


\paragraph{Bounds for No Collision Detection.}
Here we build on Theorem~\ref{thm:noninteractive} to prove a tight bound
on the advice required for deterministic solutions to the contention resolution problem.
The problem can be trivially solved in one round given $\log{n}$ bits of advice.
We show that as we reduce this advice to 
$\alpha\log{n}$ bits, for some  positive fraction $\alpha < 1$,
the rounds required grow roughly as $n^{1-\alpha}$.
We will then show this result is tight within constant factors.

\begin{theorem}
Assume deterministic algorithm $A$ and advice function $f$ solve contention resolution
in $t(n)$ rounds in a network of size $n$ with no collision detection,
such that $f$ returns at most $\alpha\log{n}$ bits
of advice, for some positive constant $\alpha < 1$.
It follows that $t(n) \geq n^{1-\alpha}/2$.
\label{thm:detlower:nocd}
\end{theorem}
\begin{proof}
Fix some $A$, $f$, $t$, $n$, and $\alpha$, as specified.
%
Our strategy is to deploy the advice function $f$, designed for general
contention resolution without collision detection, to create a related advice function $f'$
for the non-interactive setting.
To do so, given a participant set $P \subseteq V$, consider the execution 
that results when the participants in $P$ run $A$ given initial advice $f(P)$.
By assumption, this execution
solves contention resolution by round $r \leq t(n)$.
Crucially, because $A$ assumes no collision detection, no participant in this execution receives
or detects anything until round $r$. 
Therefore, each participant can simulate the first $r-1$ rounds
locally by simply simulating $A$ detecting silence for $r-1$  rounds.

It follows that to solve non-interaction contention resolution it is enough
for $f'(P)$
return the pair $(f(P), r)$, allowing each node in $P$ to locally simulate $A$
with advice $f(P)$
through the first $r-1$ rounds, and then simply execute round $r$, solving contention resolution
in one round without any interaction.
To encode this advice requires $f'$ to return
$f(P) + \lceil \log{t(n)} \rceil  =  \alpha\log{n} + \lceil \log t(n) \rceil$ bits.
By Theorem~\ref{thm:noninteractive}, this sum must be at least $\log{n}$.
It follows:

\begin{eqnarray*}
\alpha\log{n} + \lceil \log(t(n)) \rceil & \geq & \log{n}  \Rightarrow \\
\lceil \log(t(n)) \rceil & \geq & \log{n} - \alpha\log{n} \Rightarrow \\
 \log(t(n)) + 1 & \geq & (1-\alpha)\log{n}  \Rightarrow \\
2^{ \log(t(n))}2 & \geq & 2^{(1-\alpha)\log{n}} \Rightarrow \\
2t(n) & \geq & n^{1-\alpha} \Rightarrow \\
t(n) \geq & n^{1-\alpha}/2, \\
\end{eqnarray*}

\noindent  as claimed by our theorem.
\end{proof}

It is straightforward to show that
this lower bound is tight (within constant factors).
Assume we are restricted to $q = \alpha\log{n}$ bits of advice for some positive $\alpha<1$.
Assign the $n$ possible participants to the leaves of a balanced binary tree of height $\lceil \log{n} \rceil$.
Given $\log{n}$ bits of advice, it is possible to encode a traversal 
through this tree from the root down to a specific leaf, corresponding the id of a participant.
For $q<\log{n}$ bits,
we can record the first $q$ steps of such a traversal toward a participant.
Given this traversal prefix, all participants can reduce the number of possible identities
of the target participant by a factor of $2$ with each step deeper into the tree.
The result is that only $n/q = n/2^{\alpha\log{n}} = n/n^{\alpha} = n^{1-\alpha}$
possible identities remain. The participants can solve the problem in an additional
$n^{1-\alpha}$ rounds by giving each of these remaining identities one round to transmit
alone. The resulting algorithm is within a constant factor of the bound proved above.

 \paragraph{Bounds for Collision Detection.}
We now consider the collision detection case.
As mentioned, with no advice, $\Theta(\log{n})$ is a tight bound
on contention resolution with collision detection.
Here we prove that $b$ bits of advice can improve this by at most an additive
factor of $b$.

\begin{theorem}
Assume deterministic algorithm $A$ and advice function $f$ solve contention resolution
in $t(n)$ rounds in a network of size $n$ with  collision detection,
such that $f$ returns at most $b(n)$ bits
of advice.
It follows that $t(n) \geq \log{n} - b(n)$.
\label{thm:detlower:cd}
\end{theorem}
\begin{proof}
Fix some $A$, $f$, $t$, $b$ and $n$ as specified.
We use these tools to build an advice function $f'$ to solve non-interactive contention resolution.
In particular, given a participant set $P \subseteq V$,
$f'(P)$ can include the advice $f(P)$ as well as the round $r \leq t(n)$
during which these participants running $A$ with advice $f(P)$  solve contention resolution.
Unlike in the no collision detection case, however, this advice is not sufficient on its
own for participants to simulate $A$'s behavior during round $r$,
as the participants also need to know the {\em pattern} of collisions and no collisions
that occur in the first $r-1$ rounds of this simulation.

Accordingly, we have $f'(P)$ instead return $f(P)$ plus a bit string of length
$r-1$ that encodes the collision history (e.g., with $0$=silence and $1$=collision) of
the first $r-1$ rounds of participants in $P$ running $A$ with this advice.
This allows each participant to locally simulate these $r-1$ rounds and proceed
with the behavior of round $r$ which solves the non-interactive problem.

We know from Theorem~\ref{thm:noninteractive} that solving non-interactive contention resolution
requires at least $\log{n}$ bits of advice.
Therefore, it must be the case that

\[b(n) + t(n) \geq \log{n} \Rightarrow t(n) \geq \log{n} - b(n),\]

\noindent as claimed.
\end{proof}

To show this bound tight,
we first note that
a standard solution for contention resolution with collision detection is
to construct a balanced binary tree of height $\lceil \log{n} \rceil$ with
one potential participant assigned to each leaf. Using the collision detector,
participants can traverse from the root to the leaf corresponding to one of the active
participants (using silence/collision to indicate a vote for descending
to the left/right subtree). 
An obvious way to deploy $b(n)$ bits of advice to augment
this algorithm is to provide the participants
the first $b(n)$ steps of the traversal toward an active participant,
leaving them with only $\lceil \log{n} \rceil - b(n)$ steps remaining to 
complete the identification of a single node to transmit.
Resolving the impact of the ceilings provides an almost exactly
matching upper bound 
of $\log{n} - b(n) + 1$ rounds.

\subsection{Tight Bounds for Randomized Algorithms}

We now consider tight bounds on the maximum possible
speed up to the expected
time complexity of randomized algorithms given a bounded
amount  of advice.
These results deploy two styles of reduction arguments.
An interesting property of the lower bound results that follow is that because
we ultimately reduce to a lower bound that holds for non-uniform algorithms
(from~\cite{newport2014radio}), these new bounds also hold for non-uniform algorithms.

\begin{theorem}
Fix a network size $n$ and advice size $b(n) < \log{n}$.
Let $t(n)$ be the smallest bound such that an algorithm augmented
with an advice function that provides $b(n)$ bits of advice can solve
contention resolution with no collision detection in $t(n)$ rounds in expectation.
It follows that $t(n) = \Theta(\log(n)/2^{b(n)})$.
\label{thm:randlower:nocd}
\end{theorem}
\begin{proof}
In~\cite{newport2014radio}, it is shown that there exists a positive constant $c$
such that $c\log{n}$
rounds are needed to solve contention resolution in expectation with no collision detection.
If there exists an algorithm that solves contention resolution
with $b(n)$ bits of advice in $t(n) < (c\log{n})/2^{b(n)}$ rounds in expectation,
we could use it to solve contention resolution with no advice in $2^{b(n)}t(n) < c\log{n}$
rounds by simply trying all $2^{b(n)}$ advice strings in parallel,
violating the bound from~\cite{newport2014radio}.

From an upper bound perspective,
this lower bound can be matched by a truncated version of 
the {\em decay} strategy~\cite{bgi}
that uses $b$ bits of advice to reduce the size of the set of $\log{n}$
geometric size ranges to test by a factor of $2^{b(n)}$.
\end{proof}

\begin{theorem}
Fix a network size $n$ and advice size $b(n) < \log\log{n}$.
Let $t(n)$ be the smallest bound such that an algorithm augmented
with an advice function that provides $b(n)$ bits of advice can solve
contention resolution with collision detection in $t(n)$ rounds in expectation.
It follows that $t(n) = \Theta(\log\log{n} - b(n))$.
\label{thm:randlower:cd}
\end{theorem}
\begin{proof}
Theorem~\ref{thm:randlower:nocd} tells us that for a given $b(n) < \log{n}$,
you need $t'(n) = \Omega(\log{n}/2^{b(n)})$ rounds to solve contention resolution without contention
resolution and $b(n)$ bits of advice.
We argue that you need  $t(n) \geq \log{t'(n)}$ rounds to solve
contention resolution with this same amount of advice and collision detection.
To prove this, we can apply the simulation strategy from~\cite{newport2014radio}
in which we simulate a collision detection algorithm without collision detection,
by trying, in parallel, all possible collision histories (one of which will
happen to be right).
It is shown in~\cite{newport2014radio} that simulating $t(n)$ rounds of a collision
detection algorithm in this 
manner requires $2^{t(n)}$ rounds.
Therefore, if $t(n)$ was less than $\log{t'(n)}$,
we would have a contention resolution solution without collision detection
that required less than $t'(n)$ rounds---contradicting Theorem~\ref{thm:randlower:nocd}.

Give the constraint that $t(n) \geq \log{t'(n)}$, it follows:

\begin{align*}
    &t(n) \geq \log\left(  (c\log{n})/2^{b(n)} \right) = \log{(c\log{n})} - b(n)\\
    &=O(\log\log{n}) -b(n).
\end{align*}
     
To match this bound we can augment the classical strategy due to Willard~\cite{willard1986log} with $b(n)$ bits of advice.
In the setting where $b(n) =0$,
this strategy conducts a binary search over the geometric
guesses on the network size (i.e., the values in $L(n) = \{ 2^i \mid i \in [\lceil \log{n} \rceil\}$), using a collision to indicate the guess is too small,
and silence to indicate the guess is too large. If each step is repeated
a sufficiently large constant number of times, 
and the geometrically decreasing error probabilities are carefully
summed, it is shown in~\cite{willard1986log} that this
search succeeds in $O(\log{(L(n))}) = O(\log\log{n})$ rounds in expectation.

Given $b(n) < \log\log{n}$ bits of advice, we can reduce the number of values in $L(n)$
to search down to a subset $L'(n)$ of size $\lceil \log{n} \rceil/2^{b(n)}$.
The same analysis as~\cite{willard1986log} still applies as we
are searching over a  subset of $L(n)$, providing an
optimized expected round complexity in:
\[O(\log{(L'(n))}) = O(\log(\log{n}/2^{b(n)})) = O(\log\log{n} - b(n)), \]

\noindent which asymptotically matches our lower bound. For $b(n) \geq \log\log{n}$,
we can directly encode the correct range and subsequently solve the problem in $O(1)$ rounds.
\end{proof}


\bibliographystyle{ACM-Reference-Format}
\bibliography{main}

\end{document}